\newif\ifPreprint
\Preprinttrue

\ifPreprint
\documentclass[journal,onecolumn]{IEEEtran}
\else
\documentclass[journal]{IEEEtran}
\fi

\usepackage[english]{babel}  
\usepackage{amsmath,amssymb} 
\usepackage{graphicx}        
\usepackage{color}           
\usepackage{varioref}        
\usepackage{xspace}          
\usepackage{url}             
\usepackage{cite}            
\usepackage{multirow}        
\usepackage{amsthm} 

\newif\ifIsIEEE
\IsIEEEtrue

\ifIsIEEE
\usepackage[caption=false]{subfig} 
\else
\usepackage{subfig} 
\fi

\newtheoremstyle{thmstyle}{1.5ex}{-1ex}{\itshape}{}{\bfseries}{.}{0.5em}{}
\theoremstyle{thmstyle}
\newtheorem{theorem}{Theorem}
\newtheorem{lemma}[theorem]{Lemma}

\newtheorem{proposition}[theorem]{Proposition}

\newtheoremstyle{defstyle}{1.5ex}{-1ex}{}{}{\bfseries}{.}{0.5em}{}
\theoremstyle{definition}

\ifIsIEEE
\AtBeginDocument{}

\fi

\newcommand\ifnull[3]{%
  \ifx\null#1%
    #2%
  \else%
    #3%
  \fi}

\newif\ifShellEscape
\ShellEscapefalse
\newcommand\CheckShellEscape{%
  \def\tmpfile{/tmp/w18-test-\the\year\the\month\the\day\the\time}%
  \immediate\write18{touch \tmpfile}%
  \IfFileExists{\tmpfile}{\ShellEscapetrue}{\ShellEscapefalse}%
  \immediate\write18{rm -f \tmpfile}%
  \renewcommand\CheckShellEscape{\relax} 
}

\newcommand{\executeiffilenewer}[3]{%
  \ifnum\pdfstrcmp{\pdffilemoddate{#1}}{\pdffilemoddate{#2}}>0%
    \CheckShellEscape\ifShellEscape%
      {\immediate\write18{#3}}%
    \else%
      \errmessage{ERROR: Importing from Inkscape with inputsvg: When using
      the 'inputsvg' command, you should enable the --shell-escape flag in
      pdflatex}%
    \fi%
  \fi%

}

\newcommand{\Remark}{\textbf{Remark:\ }}
\newcommand{\RemarkEnd}{%
  \ifIsIEEE%
    \hfill{\IEEEQEDclosed\hspace*{0.05em}}%
  \else%
    \hfill{$\blacksquare$\hspace*{0.05em}}%
  \fi%
}

\definecolor{darkgreen}{rgb}{0,0.4,0}
\definecolor{orange}{rgb}{1,0.5,0}

\newcommand\half{\tfrac 1 2}
\newcommand\defeq{\triangleq} 

\newcommand\modop{\ {\rm mod}\ } 
\newcommand{\mo}{{-1}} 

\newcommand{\ti}[1]{^{(#1)}}

\newcommand\T{^\star}
\newcommand\pos[1]{{\backslash #1 /}}
\newcommand{\rev}{\overline}
\newcommand\mtrx[1]{\begin{pmatrix}#1\end{pmatrix}}

\newcommand\F{\mathbb F\xspace} 
\newcommand\FF[1]{\mathbb F_{#1}\xspace} 
\newcommand\ratex[3][x]{\tfrac{#2(#1)}{#3(#1)}}
\newcommand\RR{\mathbb R\xspace}
\newcommand\ZZ{\mathbb Z\xspace}

\newcommand{\params}[3]{[#1,\ #2,\ #3]}

\usepackage{algorithm}
\usepackage{algpseudocode}

\usepackage{ifpdf}

\algrenewcommand\algorithmicrequire{\textbf{Input:}}
\algrenewcommand\algorithmicensure{\textbf{Output:}}
\newcommand{\Fail}{\ensuremath{\mathsf{Fail}}\xspace}

\renewcommand{\deg}{\Delta}
\newcommand{\LTord}[2]{\textnormal{\footnotesize LT}_{#1}(#2)}
\newcommand{\xdeg}{\deg^x}
\newcommand{\ydeg}{\deg^y}
\newcommand{\wdeg}[1]{\deg_{(#1)}}
\renewcommand\pos[1]{{\mathrm{pos}(#1)}}

\def\eps{\varepsilon}
\newcommand{\errs}{\epsilon}
\newcommand{\C}{\mathcal C}

\renewcommand{\Remark}{\textbf{Remark:\ }}
\renewcommand{\RemarkEnd}{%
  \ifIsIEEE%
    \hfill{\IEEEQEDclosed\hspace*{0.05em}}%
  \else%
    \hfill{$\blacksquare$\hspace*{0.05em}}%
  \fi%
}

\begin{document}

\title{On Rational-Interpolation Based List-Decoding and List-Decoding Binary Goppa Codes}

\author{Peter~Beelen,
        Tom~H\o{}holdt,~\IEEEmembership{Fellow,~IEEE,}
        Johan~S.~R.~Nielsen,
        and~Yingquan~Wu,~\IEEEmembership{Senior Member,~IEEE}
        \thanks{%
          \ifPreprint%
            Preprint.
            This paper is published in IEEE Transactions of Information Theory with DOI: 10.1109\slash TIT.2013.2243800.

          \fi%
          The authors gratefully acknowledge the support from the Danish National Research Foundation and the National Science Foundation of China (Grant No.11061130539) for the Danish-Chinese Center for Applications of Algebraic Geometry in Coding Theory and Cryptography.

          P.~Beelen, T.~H\o{}holdt, J.~S.~R.~Nielsen are with the Department of Applied Mathematics and Computer Science, Technical University of Denmark (e-mail: pabe@dtu.dk; tomh@dtu.dk; jsrn@jsrn.dk).
          Y.~Wu is with Sandforce Inc., Milpitas, CA USA (e-mail: yingquan\_wu@yahoo.com).
        }
}

\maketitle

\begin{abstract}
We derive the Wu list-decoding algorithm for Generalised Reed-Solomon (GRS)
codes
by using Gr\"obner bases over modules and the Euclidean algorithm (EA) as the
initial algorithm instead of the Berlekamp-Massey algorithm (BMA). We present a
novel method for constructing the interpolation polynomial fast. We give a new
application of the Wu list decoder by decoding irreducible binary Goppa codes
up to the binary Johnson radius. Finally, we point out a connection between the
governing equations of the Wu algorithm and the Guruswami-Sudan algorithm
(GSA), immediately leading to equality in the decoding range and a duality in
the choice of parameters needed for decoding, both in the case of GRS codes and
in the case of Goppa codes.
%
\end{abstract}


\section{Introduction}

\IEEEPARstart{I}{n} \cite{wu08}, Wu presented a decoding algorithm for
Generalised Reed-Solomon
(GRS) codes which decodes beyond half the minimum distance. Just like the
Guruswami-Sudan algorithm (GSA) \cite{guruSudan99}, the decoder might return a
list
of candidate codewords, justifying the term \emph{list decoder}. The two
algorithms share many other properties, most notably the decoding radius: they
can both decode an $\params n k {n-k+1}$ GRS code up to $n - \sqrt{n(k-1)}$; the
so-called Johnson radius.

The Wu list decoder reuses the output of the Berlekamp-Massey algorithm (BMA).
The BMA has long been used for solving the Key Equation of GRS codes
\cite{berlekamp} whenever the number of errors is less than half the minimum
distance. Wu noted that the result of the BMA still reveals crucial information
about solutions to the Key Equation when more errors have occurred, and used
this for setting up a rational interpolation problem. This problem can be
solved by a generalisation of the core of the GSA, which
solves a similar problem for polynomials.

The equivalence of the BMA and a special utilisation of the extended Euclidean
algorithm (EA) is well-studied, e.g.~ \cite{fitzpatrick95, dornstetter87,
heydtmann00}. Inspired especially by Fitzpatrick \cite{fitzpatrick95}, we
recast the Key Equation and the first part of the Wu list decoder into the
language of Gr\"obner bases over certain modules, making it possible to use the
EA; a generally more flexible and algebraic approach than the BMA.

The rational interpolation problem is attacked by first constructing an
interpolation polynomial. This can be done by solving a large linear system of
equations, but that is prohibitively slow. We give a fast method for
constructing the interpolation polynomial which has the same asymptotic
complexity as the fastest known methods for polynomial interpolation as used in
the GSA. This also renders the Wu list decoder as fast as the fastest variants
of the GSA.

The decoding radius and the choice of auxiliary parameters in the Wu list
decoder is governed by having to satisfy a certain inequality, just as in the
GSA; we point out that in the case of decoding GRS codes, the inequality in the
Wu list decoder \emph{becomes} the governing inequality by a simple change of
variables, immediately implying that they have the same decoding radius and
always use the same list size.

We show how the Wu list decoder can be adapted to decode binary Goppa codes.
The algorithm is a continuation of the Patterson decoder \cite{patterson75},
and the adaption of the Wu list decoder to this case is particularly simple due
to the use of the EA instead of the BMA. Similarly to the case of GRS codes, we
point out a connection between the governing inequality of the decoding
parameters and the equation for the GSA with the K\"otter-Vardy multiplicity
assignment method (GSA+KV). This immediately yields that the methods have the
same decoding radii, namely up to the binary Johnson radius $\half n - \half
\sqrt{n(n-2d)}$, where $n$ is the length and $d$ the designed minimum distance
of the Goppa code. Using our fast interpolation method, also this algorithm is
as fast or faster than the previously known algorithms with the same decoding
radius. 

\subsection{Related Work}

The Wu list decoder is fairly recent and not much work has been done on it yet.
In both Trifonov \cite{trifonov12, trifonovWu10} and Ali and Kuijper
\cite{ali11}, an algorithm very closely related to the Wu list decoder for
GRS codes is reached using a Gr\"obner basis description. The algorithm,
however, revolves around two polynomials $G(x)$ and $R(x)$, where $G(x)$ is
defined as the polynomial vanishing at the evaluation points of the code and
$R(x)$ is the Langrange polynomial through the received word coordinates at the
evaluation points. These polynomials are of higher degree than those used by
the original Wu list decoder: the syndrome polynomial and a ``modulus''
$x^{n-k}$. More importantly, they are quite specific to the setting of decoding
GRS codes.

We take a slightly different approach, closer to the original one by Wu. We
essentially show how rational interpolation can help in solving Key
Equations; that is, equations of the form
\[
  \gamma(x)q(x) \equiv \delta(x) \mod p(x)
\]
where $p, q$ are known polynomials, and one seeks $\gamma$ and $\delta$ of low
degree while additionally having certain knowledge on the evaluations of
$\gamma$ and $\delta$. In the special case of GRS codes, this is exactly
what the Wu list decoder does, but our description also immediately makes it
clear that this can be used for binary Goppa codes.

The construction of the interpolation polynomial in the GSA is one of the most
computationally expensive parts of that algorithm. A fast method for this is by
Beelen and Brander \cite{beelenBrander} which refines one by Lee and O'Sullivan
\cite{leeOSullivan08}; the main gain comes from solving the core
polynomial-matrix problem using a faster method by Alekhnovich
\cite{alekhnovich05}. There is an
even faster method for this matrix problem by Giorgi et al. \cite{giorgi03},
and using this in \cite{beelenBrander} yields the fastest known way of
constructing the interpolation polynomial.
Bernstein uses essentially the same approach for his GSA variant and achieves
the same speed \cite{bernstein11simplified}, both for Reed-Solomon codes and
alternant codes; see also below. We show how this approach can be extended for
rational interpolation, which ultimately leads to the Wu list decoder having
the same asymptotic complexity as the GSA.

Binary Goppa codes have long been known to have much better minimum distances
than their underlying GRS codes: if constructed with Goppa polynomial of
degree $t$, the minimum distance is at least $2t + 1$, while it's GRS code has
minimum distance $t+1$, see e.g.~\cite{MacWilliamsSloane}. Patterson's classic
decoding algorithm utilises the binary property to decode $t$ errors
\cite{patterson75}, but recent advances in list decoding allows decoding up to
the binary Johnson radius $J_2 = \half n - \half\sqrt{n(n-4t -
2)} > t$, where $n$ is the length of the code.

Simply list decoding the underlying GRS code only reaches $n - \sqrt{n(n-t-1)}
< t$, so this is not sufficient. However, by considering the Goppa code as one
constructed with a degree $2t$ Goppa polynomial by utilising the identity of
\cite{skhn76}, and then using the GSA+KV, one reaches $J_2$, see
e.g.~\cite{augot10} or \cite[Section 9.6]{rothBook}. Alternatively, one can
with the identity
of \cite{skhn76} use Bernstein's decoder for alternant codes which works in a
manner closely related to the GSA+KV \cite{bernstein11simplified}.

The K\"otter-Vardy method does not directly translate to the Wu list decoder, so
a different approach is required. Our algorithm continues the original
insights by Patterson by rewriting the Key Equation of the Goppa code into a
reduced one of only half the degrees. This combined with list decoding turns
out to also reach $J_2$.

\subsection{Organisation}

The remainder of this article is organised as follows: The introduction ends
with
some notation and notes on the modules that will be considered. In Section
\ref{sec:ea} we describe how solutions to certain Key Equation-like equations
can be described using these modules, and how the EA can find these. In Section
\ref{sec:ratInter}, we introduce the problem of rational interpolation as well
as a method to solve it for some parameters. We then show how the solution of
the rational interpolation problem can be computed with low
complexity. These two
theoretical sections are then utilised in sections \ref{sec:grs} and
\ref{sec:goppa} for decoding GRS codes and binary Goppa codes respectively. For
each of those code families, we analyse the parameters needed for solving the
associated rational interpolation problem, and we compare asymptotic running
times with previous decoding methods.

\subsection{Notation}

Let $\F$ be a finite field. Define $R \subset \F[x,y]$ as all bivariate
polynomials over $\F$ with $y$-degree at most 1. In this article, we will be
considering $\F[x]$-modules that are subsets of $R$. Such a module could just
as well be regarded as a subset of $\F[x] \times \F[x]$; however, using
bivariate polynomials does give certain notational advantages.

We can define term orders as well as Gr\"obner bases over such modules. These
definitions follow the general intuition from Gr\"obner bases over polynomial
ideals. For an extensive presentation, see e.g.~\cite{coxUsingAlgGeo}.

One thing to keep in mind is that term orders over $\F[x]$ sub-modules of $R$
differs slightly from term orders over the polynomial ring $\F[x,y]$. For
instance, the weighted degree term order giving $x$ weight 1 and $y$ weight 0,
as well as lexicographically ordering $x$ before $y$, is a valid module term
order for these modules, while it is not valid over $\F[x,y]$.

For our discussions on modules and term orders, we define the following
notational short-hands, where $<$ is a module term order and $h(x,y) \in R$:
\begin{itemize}
\item $[ h_1, \ldots, h_t ] \defeq
      \big\{ \sum_{i=1}^t a_i(x)h_i(x,y) \mid a_i(x) \in \F[x] \big\}$ is
      the $\F[x]$-module generated by $h_1,\ldots,h_t \in \F[x,y]$.
\item $\deg f    \defeq \textrm{deg}~f(x)$ for $f(x) \in \F[x]$. Also
      define $\deg f= -\infty$ when $f(x) = 0$.
\item $\LTord < h$ is the leading term of $h$ wrt.~$<$.
\item $\xdeg_<(h) \defeq x\textrm{deg}(\LTord < h)$, where
  $x\textrm{deg}(x^iy^j)=i$.
\item $\ydeg_<(h) \defeq y\textrm{deg}(\LTord < h)$, where
  $y\textrm{deg}(x^iy^j)=j$.
\end{itemize}
Note in particular here that the $\ydeg_<(h)$ of an $h \in R$ is \emph{not}
the usual $y$-degree of $h$, but instead the $y$-degree of its leading term.
In a sense, it describes the \emph{position} of the leading term in $h$.

\section{The Euclidean algorithm and Gr\"obner bases}\label{sec:ea}
\begingroup \newcommand{\lr}{{<_\mu}}

Consider the following problem generalised from the Key Equation of algebraic
coding theory: we are given two polynomials $p(x), q(x)$, and we seek two other
polynomials $\gamma(x), \delta(x)$ of relatively low degrees which satisfy
\begin{equation}\label{eqn:genKeyEq}
  \gamma(x)q(x) \equiv \delta(x) \mod p(x)
\end{equation}
This equation alone might not be sufficient to uniquely determine $\gamma(x)$
and $\delta(x)$, but we would still like to gather as much information from the
above equation as possible, in a certain sense.
 
Consider now the set $M = [ p(x), y-q(x) ] \in \F[x,y]$ as a module over
$\F[x]$. We easily see that the polynomial $\delta(x) - y\gamma(x)$ is in $M$
by using the above congruence:
\begin{align*}
  \delta(x) - y\gamma(x) &= (\gamma(x)q(x) - w(x)p(x)) - y\gamma(x) \\
               &= -\gamma(x)(y-q(x)) - w(x)p(x)
\end{align*}
for some polynomial $w(x)$. We might therefore study $M$ in order to get a good
description of $\gamma(x)$ and $\delta(x)$; we could, for example, seek a basis
for $M$ in which $\delta(x) - y\gamma(x)$ described in this basis has
coefficients of low degree. As we will see, this can be given by a Gr\"obner
basis under a certain module term order.

For a given ordering, we have the following easy condition for a generating set
to be a Gr\"obner basis for the considered type of modules:
\begin{proposition}\label{prop:moduleGrob}
    Let $M = [p(x),y-q(x)]$ be a module over $\F[x]$ for two polynomials
    $p(x),q(x)$ and let $<$ be a module term order. A set $G = \{ h_1(x,y),
    h_2(x,y)
    \}$ is a Gr\"obner basis of $M$ under $<$ if and only if $[G] = M$ and
    $\ydeg_<(h_1) \neq \ydeg_<(h_2)$.
\end{proposition}
\begin{proof}
  Follows straight-forwardly by applying Buchberger's $S$-criterion.
\end{proof}

For any $\mu \geq 0$, define now the module term order $\lr$ as the $(1,\mu)$
weighted-degree ordering of $(x,y)$ with $x > y$. For example, $x^{\mu-1} \lr\
y\ \lr x^\mu$. We can now characterise the form of a Gr\"obner basis for $M$
under
this module term order, as well as the form of $\delta(x) - y\gamma(x)$ in
this basis, given a limit on the degree of $\gamma$:
\begin{proposition}\label{prop:grobDegs}
Let $G =  \{h_1(x,y), h_2(x,y) \}$ be a Gr\"obner basis for $M = [p(x),y-q(x)]$
under $\lr$
with $\ydeg_\lr(h_1) = 0$. Then $\xdeg_\lr(h_1) + \xdeg_\lr(h_2) = \deg
p$.\\
Furthermore, if $\delta(x) - y\gamma(x) \in M$, then there exist polynomials
$f_1(x), f_2(x)$ such that
\[
  \delta(x) - y\gamma(x) = f_1(x)h_1(x,y) + f_2(x)h_2(x,y)
\]
If $\delta(x) <_\mu y\gamma(x)$ then these polynomials satisfy
\begin{IEEEeqnarray*}{rCl}
          \deg f_1 &\leq& \deg \gamma + \mu - \xdeg_\lr(h_1) - 1\\
          \deg f_2 &=& \deg \gamma - \xdeg_\lr(h_2)
\end{IEEEeqnarray*}
If $\delta(x) >_\mu y\gamma(x)$ then they instead satisfy
\begin{IEEEeqnarray*}{rCl}
    \deg f_1 &=& \deg \delta - \xdeg_\lr(h_1) \\
    \deg f_2 &\leq& \deg \delta - \mu - \xdeg_\lr(h_2)
\end{IEEEeqnarray*}
\end{proposition}
\begin{proof}
Let us first prove the degree bounds on $h_1$ and $h_2$. Write $h_1(x,y) =
h_{10}(x)+yh_{11}(x)$ and $h_2(x,y) = h_{20}(x)+yh_{21}(x)$. Note that
$h_{11}(x)$ and $h_{21}(x)$ are coprime since some linear combination of them
gives $1$, as $y - q(x) \in M$. Then $f(x) = h_{21}(x)h_1(x,y) -
h_{11}(x)h_2(x,y) \in M$ and does not contain $y$, and is the lowest degree
polynomial in $M$ to do so; this must be $cp(x)$ for some $c \in \F$, given the
definition of $M$. Therefore $\deg f = \deg p$. However, by expanding the
expression for $f$, we get
\begin{align*}
    \deg f  &= \deg \big( h_{21}(x)h_{10}(x) - h_{11}(x)h_{20}(x) \big) \\
            &= \deg (h_{21}(x)h_{10}(x) )
\end{align*}
where we have used $\ydeg_\lr h_1 = 0$ and $\ydeg_\lr h_2 = 1$, the latter
implied by Proposition \ref{prop:moduleGrob}.

Now for the statement on $\delta(x) - y\gamma(x)$. It is clear that $f_1, f_2$
satisfying the first of the equations exist, but we need to show the degree
bounds. Assume first $\delta(x) <_\mu y\gamma(x)$. $f_1, f_2$ can be found by
the division algorithm, so we consider how this would run. As $\delta(x) <_\mu
y\gamma(x)$, we know that $h_2$ will be used as a divisor first, and it will
divide so as to cancel the leading term; this first division therefore
determines the degree of $f_2$ to be $\deg \gamma - \deg h_{21} = \deg \gamma -
\deg_\lr^x h_2$. We might then perform more divisions by $h_2$ until at one
point we use $h_1$; by then the remainder will be reduced to some $\grave
\delta(x) - y\grave \gamma(x)$ with also $\grave \delta(x) >_\mu y\grave
\gamma(x)$, and this division then determines the maximal degree of $f_1$ to
$\deg \grave \delta - \deg h_{10}$. The division algorithm ensures us that the
iterations has ``decreased'' the remainder, i.e. $\grave \delta(x) - y\grave
\gamma(x) <_\mu \delta(x) - y\gamma(x)$ and therefore $\grave \delta(x) <_\mu
y\gamma(x)$.
As $\lr$ lexicographically orders $x$ before $y$, we therefore must have $\deg
\grave \delta(x) \leq \deg \gamma + \mu - 1$. In all, we get $\deg f_1 \leq \deg
\gamma + \mu - \deg_\lr^x h_1 - 1$. The case $\delta(x) >_\mu y\gamma(x)$ runs
similarly.
\end{proof}

It turns out that the EA, if running on $p(x)$ and $q(x)$, in a certain manner
produces Gr\"obner bases of the module $M$ of module term order $<_\mu$ . To
prove
this, we first need to remind of well-known results on the intermediate
polynomials computed by the algorithm. For brevity, we don't present the EA
algorithm in full, and consequently we can't prove the following lemma, but
there are many good expositions on the algorithm which includes these results,
e.g. Tilborg \cite[Lemma 4.5.4]{tilborgCoding} or Dornstetter
\cite{dornstetter87}.

Consider running the Extended Euclidean Algorithm (EA) on $p(x)$ and $q(x)$,
and denote by $s_i(x)$ the remainder polynomial computed in each iteration $i$;
that is, $s_0(x) = p(x)$, $s_1(x) = q(x)$ and $s_2(x), s_3(x), \ldots, s_N(x),
s_{N+1}(x)$
will be the following remainders computed, where we know by the EA that $s_N(x)
= \gcd(p,q)$ and $s_{N+1}=0$. Then the EA in each iteration $i$ also computes
polynomials $u_i(x), v_i(x)$ such that $s_i(x) = u_i(x)p(x) + v_i(x)q(x)$.
Furthermore, we have the following lemma, whose proof is easy by induction
on the precise computations of the EA:
\begin{lemma}\label{lem:EAprops}
If the EA is run on polynomials $p(x), q(x)$ with $\deg p > \deg q$, the
intermediate polynomials satisfy for each iteration $i=1,\ldots,N+1$:
\newcounter{itemcounter}
\def\Item{%
    \addtocounter{itemcounter}{1}%
    \quad\textrm{\upshape(\roman{itemcounter})}&%
}
\begin{IEEEeqnarray*}{-s?rCl}
\Item  \IEEEeqnarraymulticol{3}{s}{%
            $\deg s_i$ is a decreasing function in $i$.
        } \\
\Item  (-1)^i         &=& u_i(x)v_{i-1}(x) - u_{i-1}(x)v_i(x)    \\
\Item  s_i(x) &=& u_i(x)p(x) + v_i(x)q(x) \\
\Item  \deg p &=& \deg v_i + \deg s_{i-1}
\end{IEEEeqnarray*}
\end{lemma}
We are now in a position to show how each iteration of the EA gives rise to a
generating set for $M$:
\begin{proposition}\label{prop:EAgens}
Let the EA be run on two polynomials $p(x),q(x)$ with $\deg p > \deg q$. In
each iteration $i$, let $G = \{ h_1(x,y), h_2(x,y) \}$ with
\begin{align*}
  h_1(x,y)  &= s_{i-1}(x) - v_{i-1}(x)y \\
  h_2(x,y)  &= s_i(x) -   v_i(x)y
\end{align*}
Then $[G] = M$ where $M = [p(x), y-q(x)]$.
\end{proposition}
\begin{proof}
  Inserting the expression for $s_i(x)$ and $s_{i-1}$ from Lemma
  \ref{lem:EAprops} (iii), we get
  \begin{IEEEeqnarray*}{rCl}
    \mtrx{ h_1(x,y) \\ h_2(x,y) }
      &=& \mtrx{ u_{i-1}(x) & -v_{i-1}(x) \\ u_i(x) & -v_i(x)}
    \mtrx{ p(x) \\ y-q(x) }
  \end{IEEEeqnarray*}
  Now $h_1(x,y), h_2(x,y)$ and $p(x), y-q(x)$ will be bases for the same module
  if and only if the determinant of the $2\times2$-matrix is a unit. But this
  is stated in Lemma \ref{lem:EAprops} (ii).
\end{proof}

We can now wrap up and show the main result of this section:
\begin{proposition}\label{prop:EAgrob}
Let $p(x),q(x)$ be two polynomials with $\deg p > \deg q$, and let $\mu \geq 0$
be an integer. If the EA is run on $p(x), q(x)$ and it is halted on the first
iteration
$i$ where $\deg s_i < \deg v_i + \mu$, then $G = \{ h_1(x,y), h_2(x,y) \}$ is a
Gr\"obner basis of $M = [p(x), y-q(x)]$ with module term order $<_\mu$, where
$h_1, h_2$ are chosen as in Proposition \ref{prop:EAgens} for iteration $i$.
\end{proposition}
\begin{proof}
Clearly there \emph{is} a first iteration $i$ where $\deg s_i < \deg v_i +
\mu$, for $\deg s_{N+1} = -\infty$ and $\deg v_{N+1} \geq 0$. Thus, at least
the $(N+1)$st iteration satisfies the requirement. Conversely, the $0$'th
iteration does not satisfy it as $\deg s_0 = \deg p$ and $\deg v_0 = -\infty$.
Now to show that $G$ is a Gr\"obner basis. From Proposition \ref{prop:EAgens}
we
know that $[G] = M$, so by Proposition \ref{prop:moduleGrob} we only need to
show that the leading terms of $h_1$ and $h_2$ have different $y$-degree under
$<_\mu$. But by the choice of $i$, we have both $\ydeg_\lr(h_1) = 0$ and
$\ydeg_\lr(h_2) = 1$.
\end{proof}

\endgroup 
\section{Rational Interpolation}\label{sec:ratInter}

We will now describe how to solve the problem of finding rational curves that
go through at least some number of prescribed points. The method is a
generalisation of the GSA \cite{guruSudan99}, and first
described by Wu \cite{wu08}. The formulation of our main theorem, Theorem
\ref{thm:interpol}, avoids some special handling of points at infinity and is
due to Trifonov \cite{trifonovWu10}.

We are basically interested in a rational expression $\ratex{f_2}{f_1}$ with
numerator and denominator of low degrees, which goes through at least some
$\tau$ out of $n$ points $\big( (x_0, \beta_0), \ldots, (x_{n-1},
\beta_{n-1})\big)$ where all $x_i \in \F$ while $\beta_i \in \F \cup
\{\infty\}$. To handle the points at infinity, we can instead consider these as
partially projective points $(x_i, y_i: z_i)$ with $\frac {y_i}{z_i} = \beta_i$
whenever $\beta_i \neq \infty$ and $(y_i, z_i) = (1,0)$ otherwise.

In this language, the interpolation amounts to finding low-degree polynomials
$f_1(x)$ and $f_2(x)$ such that for at least $\tau$ values of $i$, we have
$y_if_1(x_i) - z_if_2(x_i) = 0$. The following theorem is a paraphrasing of
\cite[Lemma 3]{trifonovWu10}; we omit the proof which is a
generalisation of the proof of \cite[Lemma 4]{guruSudan99}.

First a notational short-hand: For a $Q \in \F[x,y,z]$, we define
\begin{IEEEeqnarray*}{l}
\wdeg{w_x,w_y,w_z}Q(x,y,z) \defeq \max\{ iw_x+jw_y+hw_z  \\
  \qquad\qquad\qquad \mid \alpha x^iy^jz^h\ \textrm{is a monomial
           of } Q(x,y,z) \}
\end{IEEEeqnarray*}
That is, $\wdeg{w_x,w_y,w_z} Q(x,y,z)$ is the $(w_x,w_y,w_z)$-weighted degree
of $Q$. Now the theorem:

\begin{theorem}\label{thm:interpol}
Let $\ell, s$ and $\tau$ be positive integers, and let $\{(x_0,y_0,z_0),
\ldots, (x_{n-1},y_{n-1},z_{n-1}) \}$ be $n$ points in $\F^3$ where for all $i$
either $y_i$ or $z_i$ is non-zero. Assume that $Q(x,y,z) = \sum_{i=0}^\ell
Q_i(x)y^iz^{\ell-i}$ is a non-zero partially homogeneous trivariate polynomial
such that $(x_i, y_i, z_i)$ are zeroes of multiplicity $s$ for all $i =
0,\ldots,n-1$, and $\wdeg{1,w_2,w_1}Q < s\tau$, for two $w_1, w_2 \in \RR_+
\cup \{0\}$.
Any two coprime polynomials $f_1(x), f_2(x)$ satisfying $\deg f_1 \leq
w_1$, $\deg f_2 \leq w_2$, as well as $y_if_1(x_i) + z_if_2(x_i) = 0$ for at
least $\tau$ values of $i$, will satisfy $(yf_1(x) + zf_2(x)) \mid Q(x,y,z)$.
\end{theorem}

As with the GSA, such a trivariate polynomial can be constructed by
setting up and solving a system of linear equations. Each point to go through
with multiplicity $s$ corresponds to a similar requirement in a bivariate
polynomial (see e.g. \cite[Lemma 1]{trifonovWu10}), and therefore gives rise to
$\half s(s+1)$ linear equations, so the total number of equations is given by
$\half ns(s+1)$. The number of coefficients of $Q$ -- and therefore variables
of the equation system -- is at least $\sum_{i=0}^\ell s\tau-iw_2 -
(\ell-i)w_1$; it is exactly this whenever all the terms in the sum are
non-negative, but it can actually be more when some of them are negative.
Expanding and collecting, we therefore have that at least any $n, \tau, w_1,
w_2, \ell, s$ which satisfy:
\begin{equation}\label{eqn:ratPermissible}
    \half ns(s+1) < s\tau(\ell+1) - \half\ell(\ell+1) (w_1+w_2)
\end{equation}
allow for a construction of a satisfactory $Q$.

It is easy to see that $Q$ can have at most $\ell$ factors of the form given in
the theorem, as its $y$-degree is $\ell$. For this reason, particularly
inspired by its use for decoding and in concordance with the GSA, it is called
the \emph{(designed) list size}.

We are mostly interested in knowing for which values of $n, \tau$ and $w_1,w_2$
we can select $s$ and $\ell$ such that the above is satisfied. For rational
interpolation in general, a selection yielding a minimal $\ell$ is done in \cite{trifonov12}, so we will not repeat it here. In sections \ref{sec:grs} and \ref{sec:goppa} we use rational interpolation for decoding, and we will show a relation between the possible choices of parameters for these instances and
similar instances of polynomial interpolation using the GSA respectively
GSA+KV; this turns out to immediately give us bounds on $\tau$ as well as
values for $s$ and $\ell$.

Theorem \ref{thm:interpol} parallels a result for polynomial interpolation
as used in the GSA, see e.g.~\cite[Lemma 5]{guruSudan99}.
However, for the application of decoding, it is not quite enough; when we later
need to solve a rational interpolation problem for decoding, we seek $f_1$ and
$f_2$ which interpolate the error positions, and therefore an unknown number of
points, but their maximal degrees increase with the number of points they
interpolate. This means that we can't use Theorem \ref{thm:interpol} directly:
setting $\tau$ low while the allowed degrees of $f_1, f_2$ high would not allow
us to construct $Q$, while setting $\tau$ high would not guarantee that we
found $f_1$ and $f_2$ when only few points were interpolated. Luckily, we have
the following lemma which says that the $Q$ we construct for high $\tau$ will
also find $f_1$ and $f_2$ that interpolate fewer points, as long as their
degrees decrease appropriately:
\begin{lemma}\label{lem:risingRat}
  Let $Q(x,y,z)$ satisfy the requirements of Theorem \ref{thm:interpol} for
  some $(\tau, \ell, s, w_1, w_2)$. Then $Q(x,y,z)$ also satisfies the
  requirements for $(\tilde \tau, \ell, s, \tilde w_1, \tilde w_2)$ as
  long as
  \vspace*{-.3em}
  \[\min\{
  w_1 - \tilde w_1 \ , \ w_2 - \tilde w_2 \} \geq \tfrac s \ell (\tau
  - \tilde \tau)
  \]
  \vspace*{-1.7em}
\end{lemma}
\begin{proof}
Since the interpolation points and multiplicity as well as the list size have not
changed, we only need to show $\wdeg{1,\tilde w_2,\tilde w_1}Q < s \tilde \tau$ We have:
\begin{IEEEeqnarray*}{rCl}
\wdeg{1,\tilde w_2,\tilde w_1}Q &\leq& \wdeg{1,w_2,w_1}Q \\
    \IEEEeqnarraymulticol{3}{r}{
        \quad - \min\{ i(w_2 - \tilde w_2) + (\ell-i)(w_1 - \tilde w_1)
            \mid 0\leq i \leq \ell \}} \\
 &<& s\tau - \ell \min\{w_1 - \tilde w_1,\ w_2 - \tilde w_2 \}
\end{IEEEeqnarray*}

Therefore $Q$ satisfies the degree constraints whenever
\begin{IEEEeqnarray*}{rCl+c}
 s\tau - \ell \min\{w_1 - \tilde w_1,\ w_2-\tilde w_2 \} &\leq& s\tilde \tau
    & \iff \\
 \min\{w_1 - \tilde w_1,\ w_2 - \tilde w_2 \} &\geq&
    \frac s \ell (\tau - \tilde \tau)
\vspace*{-2.2\parskip}
\end{IEEEeqnarray*}
\end{proof}
\vspace*{-\parskip}

\subsection{Fast interpolation}\label{sec:ratFast}

As mentioned, the interpolation polynomial $Q(x,y,z)$ can be constructed by
setting up and solving a linear system of equations. However, without more
thought, this would have a cubic running time in the number of system equations, which is prohibitively slow. In this section, we describe a fast way to
construct the polynomial, building heavily upon ideas from the similar problem
in the GSA, in particular Lee and O'Sullivan \cite{leeOSullivan08} and the
subsequent refinement in Beelen and Brander \cite{beelenBrander}.

In the context of Theorem \ref{thm:interpol}, consider given values of the
parameters. We will assume that $\ell \geq s$; in later sections where we apply
rational interpolation, this turns out always to be the case. Consider now the
set $W \subset \F[x,y,z]$ consisting of \emph{all} polynomials homogeneous of
degree $\ell$ in $y$ and $z$, and which interpolate the $n$ points
$\{(x_0,y_0,z_0), \ldots, (x_{n-1}, y_{n-1}, z_{n-1})\}$, each with
multiplicity at least $s$. Our goal is then to find a non-zero $Q \in W$ of
lowest
possible $(1,w_2,w_1)$-weighted degree. It is easy to see that $W$ is an
$\F[x]$-module. The
approach is to give an explicit basis for $W$, represent this basis as a matrix
over $\F[x]$ and then use an off-the-shelf algorithm for finding the
``shortest'' vector in that matrix, ``short'' being defined appropriately. This
will correspond to a satisfactory interpolation polynomial.

Let us assume without loss of generality that each $z_i \in \{0,1\}$ and define $L =  \{ x_i | z_i=0 \}$. Define
the following polynomials which will turn out to play a crucial role: $R_y(x)$
and $R_z(x)$ will be the Lagrange polynomials interpolating $(x_i, y_i)$
respectively $(x_i, z_i)$, $i=0,\ldots,n-1$. Define also $G(x) =
\prod_{i=0}^{n-1}(x-x_i)$ as well as $g_z(x) = \gcd(G, R_z) = \prod_{i\in L}(x-x_i)$. Now, there must
exist $\lambda_1(x),\lambda_2(x) \in \F[x]$ such that $g_z(x) =
\lambda_1(x)G(x) + \lambda_2(x)R_z(x)$. Define $\Upsilon(x) = \big(
\lambda_2(x)R_y(x) \modop G(x) \big)$, considered in $\F[x]$. Note that
$\Upsilon(x_i) = \lambda_2(x_i)y_i$ for all $i=0,\ldots,n-1$. We begin with a
small lemma:

\begin{lemma}\label{lem:fastIntergz}
  Let $P(x,y,z) \in W$ and $P(x,y,z) = \sum_{j=0}^\ell P_j(x)y^jz^{\ell-j}$.
  Then $g_z(x)^{j-(\ell-s)} \mid P_j(x)$ for $j=\ell-s+1,\ldots, \ell$.
\end{lemma}
\begin{proof}
  As $P$
  interpolates the points $(x_i,y_i,z_i)$ with multiplicity $s$, $P(x+x_i,
  y+y_i, z+z_i)$ can have no monomials of total degree (in $x,y$ and $z$) less
  than $s$. For $x_i \in L$ we have $P(x+x_i, y+y_i, z+z_i) = \sum_{j=0}^\ell
  P_j(x+x_i)(y+y_i)^jz^{\ell-j}$. All the terms in the sum have different
  $z$-degree, so nothing between these terms cancels, and so each can have no
  monomials of total degree less than $s$. Multiplying out the power of $y+y_i$ reveals in particular that $P_j(x+x_i)y_i^jz^{\ell-j}$ has no monomials of degree less than $s$. Since $z_i=0$ we have
  $y_i \neq 0$, so for $j=\ell-s+1,\ldots,\ell$ we get $x^{j-(\ell-s)} \mid P_j(x+x_i)$. This implies the sought.
\end{proof}

The main result is the basis for $W$; it looks complicated, but the important
thing is that it is directly calculable given the rational interpolation
problem. We introduce for any $x\in \RR$ the function $\pos x := \max(x,0)$.
Note the easy identity $\pos x - \pos{{-x}} = x$. For the proof, we also use the
phrase ``leading monomial' of a trivariate polynomial $P(x,y,z)$ as the
monomial of highest $y$-degree when $P$ is regarded over $\F[x][y,z]$, and the
``leading coefficient'' is the $\F[x]$-coefficient of the leading monomial.

\begin{theorem}\label{thm:interpolBasis}
Let for $j=0,\ldots,\ell$
  \begin{IEEEeqnarray*}{rl}
    B\ti j\!=\ &(g_zy - \Upsilon z)^{\pos{s-j}}
              (yz - R_yz^2)^{j - \pos{j-(\ell-s)} - \pos{j-s}} \\
            &\quad (z\tfrac G {g_z})^{\pos{j-(\ell-s)}} 
              y^{\pos{\ell-s-j}} z^{\pos{j-s}}
  \end{IEEEeqnarray*}
Then $W = \big[B\ti 0, \ldots, B\ti \ell]$.
\end{theorem}
\begin{proof}
First, it should be proved that each $B\ti j$ is of total
degree $\ell$ in $y$ and $z$. By summing all the terms' exponents, counting
each $yz-R_yz^2$ twice, and using the identity for $\pos {\cdot}$ given
above, one sees this is so.

To show that each $B\ti j$ is in $W$, note first that $yz-R_yz^2$
interpolates all $(x_i,y_i,z_i)$. This is also true for $z\tfrac G {g_z}$ and
$g_zy - \Upsilon z$, since either $z_i=0$, whereby they obviously both evaluate
to 0, or $x_i \notin L$ which gives $\tfrac {G(x)}{g_z(x)}|_{x=x_i} = 0$ as
well as
\begin{IEEEeqnarray*}{rCl}
    g_z(x_i)y_i - \Upsilon(x_i)z_i &=& (\lambda_1(x_i)G(x_i)+\lambda_2(x_i))y_i
                                        - \lambda_2(x_i)y_i 
                                 \\&=& 0
\end{IEEEeqnarray*}
For each $B\ti j$ to interpolate the points with multiplicity at least $s$, we
need only to verify that the sum of the exponents of the three terms $g_zy -
\Upsilon z$, $yz-R_yz^2$ and $z\tfrac G {g_z}$ is at least $s$; this is quickly seen to be true.

We need then only to show that any $P \in W$ can be expressed as an
$\F[x]$-combination of the $B\ti j$. There are two cases to
consider, $\ell-s \leq s$ and $\ell-s > s$. We will only show the latter case,
and
the former follows similarly. So assume $\ell-s > s$. Observe that $B\ti j$ has
$y$-degree exactly $\ell-j$. The proof now basically follows the multivariate
division algorithm on $P$ under lexicographical ordering $y > z > x$;
i.e.~dividing
with the aim of lowering the $y$-degree.

First observe that the leading coefficient of $B\ti 0$ is $g_z(x)^s$. By Lemma
\ref{lem:fastIntergz}, we can perform polynomial division of $P$ by $B\ti 0$
and get a remainder $P\ti 1(x,y,z)$ of $y$-degree at most $\ell-1$. As
$B\ti 0 \in W$ so is $P\ti 1 \in W$. We can continue as such with $B\ti j$ for
$j=1,2,\ldots,s-1$, as each of these $B\ti j$ has leading coefficient
$g_z(x)^{s-j}$ and Lemma \ref{lem:fastIntergz} promises that the remainders
will keep having leading term divisible by exactly this. We thus end with a
remainder $P\ti{s}$ with $y$-degree at most $\ell-s$ and in $W$.

As $\ell-s > s$ then for $j=s,\ldots,\ell-s$ we have $B\ti j(x,y,z) =
(yz-R_yz^2)^sy^{\ell-s-j}z^{j-s}$. They all have leading coefficient $1$, so we
can reduce $P\ti s$ with $B\ti s$, reduce the remainder of that with
$B\ti{s+1}$ and so forth, until we arrive at a remainder $P\ti{\ell-s+1}$ with
$y$-degree at most $s-1$.

Still we have $P\ti{\ell-s+1} \in W$ so the $(x_i,y_i,z_i)$ are all zeroes with
multiplicity $s$. Therefore $P\ti{\ell-s+1}(x+x_i,y+y_i,z+z_i)$ has no
monomials of degree less than $s$. Let $\rev L = \{x_i | z_i \neq 0\}$ and let
$P\ti{\ell-s+1}(x,y,z) = \sum_{j=0}^{s-1}P\ti{\ell-s+1}_j(x) y^j z^{\ell-j}$.
For $x_i \in \rev L$, we see by expanding the powers of both $y+y_i$ and
$z+z_i$ that $P\ti{\ell-s+1}(x+x_i,y+y_i,z+z_i)$ has a monomial
$P\ti{\ell-s+1}_{s-1}(x+x_i) y^{s-1} z_i^{\ell-s}$ which does not cancel with
any other term. Therefore, $x \mid P\ti{\ell-s+1}_{s-1}(x+x_i) \iff (x-x_i)
\mid P\ti{\ell-s+1}_{s-1}$. Collecting for all $x_i \in \rev L$, we get $\frac
G {g_z} \mid P\ti{\ell-s+1}_{s-1}$. Note that, as $\ell-s > s$, then
$B\ti{\ell-s+1}(x)$ has leading coefficient $\frac G {g_z}$. Thus, we can
divide $P\ti{\ell-s+1}(x,y,z)$ by $B\ti{\ell-s+1}(x)$ and get remainder
$P\ti{\ell-s+2}$ of $y$-degree at most $s-2$.

Now, the exact same argument as above can be repeated for $P\ti{\ell-s+2}$, but
one finds that $(x-x_i)^2$ must divide the leading coefficient for each $x_i
\in \rev L$. Therefore, we can divide by $B\ti{\ell-s+2}$ whose leading
coefficient is $(\frac G {g_z})^2$. We can continue this way with all the
remaining $B\ti j$, until we find that the last remainder $P\ti \ell$ must be
divisible by $(\frac G {g_z})^sz^\ell = B\ti \ell$.
\end{proof}

With a concrete basis for $W$ in hand, we wish to find an element in $W$ with
lowest possible $(1,w_2,w_1)$-weighted degree.  Write the $B\ti j$ of Theorem
\ref{thm:interpolBasis} as $B\ti j(x,y,z) = \sum_{i=0}^\ell B\ti j_i(x)y^i
z^{\ell-i}$. Construct now the matrix
$\Pi \in \F[x]^{(\ell+1) \times (\ell+1)}$ where the $(j,i)$'th entry is $B\ti
j_i(x)$. The $B\ti j(x,y,z)$ thus constitute the rows of $\Pi$. In this manner,
we can represent any basis of $W$ as an $(\ell+1) \times (\ell+1)$ matrix, and
any $P \in W$ can be represented as a vector in the row span of such a basis
matrix.

Consider a vector $V$ in the row-span of $\Pi$, and denote by $|V| :=
\max_{V_j \neq 0}\{\deg V_j + jw_2+(\ell-j)w_1 \}$ where $V_j$ is the $j$'th
component of $V$. A shortest vector in $\Pi$ under this metric will
correspond to a polynomial in $W$ which has the lowest possible
$(1,w_2,w_1)$-degree. Any algorithm which can compute a shortest vector in
the row-span of an $\F[x]$-matrix under this metric will therefore be usable
to solve our problem.

The usual approach of such algorithms is to compute a so-called \emph{row
reduced} basis matrix, where the sum of the basis elements' lengths is minimal.
It is well known that the shortest vector in the row space will be present in
this reduced matrix, see e.g. \cite{lenstra85, alekhnovich05}. This problem is
widely studied and it has several different guises and names: Gr\"obner basis
reductions over free $\F[x]$-modules \cite{leeOSullivan08}, row reduction of
$\F[x]$-matrices \cite{giorgi03}, and basis reduction of $\F[x]$-lattices
\cite{mulders03}.

The fastest method in the literature for our purposes is due to Giorgi et al.
in \cite{giorgi03}. If $\theta$ is the highest degree of any polynomial in the
initial basis matrix, and the basis matrix is $\nu \times \nu$, then the
algorithm has complexity $O(\nu^\omega \theta \log^{O(1)}(\nu \theta))$, where
$O(\nu^\omega)$ is the complexity for multiplying two $\nu\times\nu$ matrices
with elements in $\F$. Trivially $\omega \leq 3$ but methods exist with $\omega
< 2.4$ \cite{coppersmith90}. To bound the running time of applying the
algorithm on our problem, we have the following:

\begin{lemma}\label{lem:PiEntries}
  In the context of Theorem \ref{thm:interpolBasis} and the discussion above,
  the entries of $\Pi$ all have degree at most $sn$.
\end{lemma}
\begin{proof}
  The entries of $\Pi$ are all of the form
  $\beta g_z^{j_1}\Upsilon^{j_2}R_y^{j_3}(\frac G {g_z})^{j_4}$ where $\beta
  \in \F$ and $j_1,j_2,j_3,j_4$ are non-negative integers summing to at most
  $s$. The lemma follows as the four base polynomials are each of degree at
  most $n$.
\end{proof}

The algorithm in \cite{giorgi03} does not directly support the different
``column weights'' that our vector metric demands, but this can be amended by
first multiplying the $j$'th column of $\Pi$ with $x^{jw_2+(\ell-j)w_1}$ and
then finding the usual row reduced basis. The powers of $x$ can then be divided
out from the resulting reduced basis afterwards. This does not change the
complexity of the algorithm whenever $w_1,w_2 \in O(n)$, which follows if we
assume $\tau^2 > n(w_1+w_2)$; an assumption which turns out to be true for our
applications in later sections. One
should also note that for finite fields $\F$, the
algorithm might need to calculate over an extension field, though without
affecting the asymptotic running time, as pointed out by Bernstein
\cite{bernstein11simplified}. This entire discussion can be distilled into the
following:

\begin{lemma}\label{lem:interpolComp}
  For given values of the parameters of Theorem \ref{thm:interpol} where
  $\ell \geq s$ and $\tau^2 > n(w_1+w_2)$, an algorithm exists to find a
  satisfactory interpolation polynomial in complexity $O(\ell^\omega s
  n\log^{O(1)}(\ell n))$.
\end{lemma}
\begin{proof}
  As soon as one has constructed $\Pi$, the result follows from Lemma
  \ref{lem:PiEntries} and the complexity of the algorithm in \cite{giorgi03},
  so we just need to show that we can compute $\Pi$ in the given speed. Let
  $M(\theta)$ be the complexity of multiplying two polynomials of degree
  $\theta$. Computing $R_y, R_z$ and $G$ by Lagrangian interpolation can be
  done in complexity $O(M(n)\log n)$, see e.g. \cite[p. 235]{gathen}.
  $\Upsilon$ and $g_z$ can be computed using the Euclidean algorithm in
  $O(n\log^2 n)$. For a polynomial of degree $n$, computing all the first
  $s$ different powers of it can be done iteratively in $O(sM(sn))$. Each entry
  in $\Pi$ is a multiple of $g_z, R_y, \frac G {g_z}$ and
  $\Upsilon$ to a combined power of $s$, so after each of their $s$ powers have
  been computed, each of the $O(\ell^2)$ entries in $\Pi$ can be computed in
  $O(M(sn))$. Using Sch\"onhage-Strassen, we can set $M(\theta) = O(\theta\log
  \theta \log\log \theta)$, see e.g. \cite[Theorem 8.23]{gathen}, and inserting
  this into the above, we see that $\Pi$ can be computed in $O(\ell^2M(sn))
  \subset O(\ell^2sn\log^{O(1)}(\ell n))$.
\end{proof}

\Remark Another algorithm that can be used to handle the interpolation problem
is the row-reduction method of Alekhnovich \cite{alekhnovich05}, which also has
been used in the interpolation method by Beelen and Brander. 
\cite{beelenBrander}. This method would, however, yield the slightly worse running time $O(\ell^{\omega+1} s n \log^{2+o(1)}(\ell
n))$.
\RemarkEnd

After having computed the interpolation polynomial $Q(x,y,z)$, one needs to
find factors of the form $yf_1(x) + zf_2(x)$ with $f_1,f_2 \in \F[x]$. Any
such factor except $z$ will also occur as an $\F(x)$ factor in the
dehomogenised version of $Q$. Thus, any fast algorithm for computing this will
suffice. In \cite{wu08}, Wu describes an extension to the root-finding method
of Roth and Ruckenstein (RRR) \cite{rothRuckenstein00} for finding $\F(x)$
roots of a $\F[x][y]$ polynomial: he remarks that simply applying the original
RRR will find the truncated power series of each $\F(x)$ root; retrieving a
long enough such series and applying a Pad\'e approximation method like the BMA
or the EA will retrieve the polynomial fraction. A divide-and-conquer speed-up
of the RRR described by Alekhnovich in \cite[Appendix]{alekhnovich05} applies just
as well to this extension\footnote{%
  We are grateful to the anonymous reviewer for pointing out the extension of
  Alekhnovich to us as well as the improvement to its running time analysis.
}. We arrive at the following:

\begin{lemma}\label{lem:interpolFactor}
  In the context of Theorem \ref{thm:interpol}, there exists an algorithm which
  finds all factors of $Q(x,y,z)$ of the form $yf_1(x) + zf_2(x)$ in complexity
  $O\big(\ell^2sn\log(\ell n)^{2+o(1)}\big)$, assuming that the cardinality of $\F$ is in $O(n)$.
\end{lemma}
\begin{proof}
  The root-finding algorithm described in \cite{wu08} will have the complexity
  of running the RRR followed by at most $\ell$ applications of the EA, each on
  a truncated power series of degree $O(\tau) \subset O(n)$. The EA applications
  will have total complexity $O(\ell n\log^2 n\log\log n)$, see e.g. \cite[Chapter 8.9]{ahoDesignCS}.

  For running the RRR, Alekhnovich reports a complexity of $O(\ell^{O(1)}\theta
  \log \theta)$, where $\theta$ is the $x$-degree of $Q(x,y,z)$; however, his
  analysis can
  be improved: in the context of his proof, choose a fast factoring method over
  $\F[y]$, e.g. from \cite[Theorem 14.14]{gathen}, and so set $f(1, \ell) =
  O(\ell M(\ell)\log(q\ell))$, where $q$ is the cardinality of $\F$. The non-recursive cost of $f(\theta, \ell)$,
  i.e. the term $\ell^{O(1)}\theta$, can be improved to $\ell^2M(\theta)$,
  as an upper bound cost of the $\ell$ different calculations of the shifts
  $Q(x,y_i+x^{d_i}\hat y)$. Now the recursive bound has the improved solution
  $f(\theta,\ell) \in O(\ell^2M(\theta)\log \theta + \theta\ell
  M(\ell)\log(q\ell))$. We have $\theta \in O(sn)$ and assume $q \in O(n)$
  and thus arrive at the complexity of the lemma.
\end{proof}

An alternative factorisation method with roughly the same complexity is
proposed by Bernstein in \cite{bernstein11simplified} by accommodating a more
classical root finding method in $\ZZ[x]$ by Zassenhaus; see also \cite[Chapter
15]{gathen}.

\section{Wu list decoding for Reed-Solomon codes}\label{sec:grs}
\begingroup \newcommand{\lz}{{<_0}}

We can now derive the Wu list decoder in a succinct manner using the Euclidean
algorithm instead of the Berlekamp-Massey algorithm (BMA). This derivation is
inspired by Trifonov's derivation \cite{trifonovWu10}, though ours is slightly
more general and uses shorter polynomials in the computations.

\subsection{The codes}

An $[n,k,d]$ Generalised Reed-Solomon (GRS) code over a finite field $\F_q$
is the set
\[
    \big\{ \big(  v_0\eta(\alpha_0),\ldots, v_{n-1}\eta(\alpha_{n-1}) \big)
        \mid \eta \in \F_q[x] \land \deg \eta < k \big\}
\]
for some $n$ distinct non-zero $\alpha_0, \ldots, \alpha_{n-1} \in \F_q$ as
well as $n$ non-zero $v_0, \ldots, v_{n-1} \in \F_q$. The $\alpha_i$ are called
\emph{evaluation points} and the $v_i$ \emph{column multipliers}. It is easy to
show that $d = n-k+1$ and the code is therefore MDS. See e.g. \cite{rothBook}
for a comprehensive introduction to GRS codes.

Consider a sent codeword $c = (c_0, \ldots, c_{n-1})$ and a corresponding
received word $r = (r_0, \ldots, r_{n-1})$. Then the syndrome polynomial is
computable by the receiver and can be defined as
\begin{equation}\label{eqn:rsSyn}
    S(x) = \sum_{i=0}^{n-k-1} x^i\sum_{j=0}^{n-1} r_j\hat v_j\alpha_j^{d-2-i}
\end{equation}
where $\hat v_j = (v_j\prod_{h \neq j}(\alpha_j-\alpha_h))^{-1}$.
If we denote the set of error locations by $E$, that is, $E = \{ i \mid c_i
\neq r_i \}$, we can define the error-locator and error-evaluator polynomials
respectively, as follows\footnote{%
The reader familiar with the three polynomials might notice our slightly
unorthodox definition of them; many sources use an error-locator which reveals
the \emph{inverse} error positions, i.e. $\Lambda(\alpha_i^\mo) = 0$ iff the
$i$'th position is in error. This also yields a slightly simpler syndrome
polynomial. However, in the case of Goppa codes, the above definition of the
error locator is more natural, and we have opted for consistency in this
article by also using that here.
}:
\begin{IEEEeqnarray*}{rCl}
    \Lambda(x) &=& \prod_{i \in E} (x-\alpha_i) \\
    \Omega(x)  &=& -\sum_{i \in E} (r_i - c_i)\alpha_i^{d-1}\hat v_i
                    \prod_{j \in E \setminus \{i\}} (x-\alpha_j)
\end{IEEEeqnarray*}

Clearly, the receiver can quickly retrieve $c$ from $r$ if he constructs
$\Lambda$ and $\Omega$, as the error locations are the roots of $\Lambda$, and
the error values are the evaluations of $\Omega$ in the respective error
location (up to a calculable scalar). Note that therefore $\gcd(\Lambda,
\Omega) = 1$ as the elements of $E$ are all the zeroes of $\Lambda$ but
definitely not zeroes of $\Omega$. The three defined polynomials are related by
the famous
Key Equation (see e.g. \cite{rothBook} or \cite{zeh11interpol}):
\begin{equation}\label{eqn:RSkey}
    \Lambda(x)S(x) \equiv \Omega(x) \mod x^{d-1}
\end{equation}
Many decoding algorithms solve this equation for $\Lambda$ and $\Omega$, and
construct $c$ from $r$ using these. That is also what our list decoder will
do.

\subsection{The list-decoding algorithm}\label{sec:RSdecode}

Using the Key Equation and the results of Section \ref{sec:ratInter}, we can
construct a list decoder. By \eqref{eqn:RSkey} as well as \eqref{eqn:genKeyEq}
on page \pageref{eqn:genKeyEq} and the paragraphs following it, we know that
$\Omega(x) -
y\Lambda(x) \in M= [x^{d-1}, y - S(x)]$. If we run the EA on $x^{d-1}$ and $S$,
by Proposition \ref{prop:EAgrob}, we get a Gr\"obner basis $G = \{h_1, h_2\}$
of $M$ of module term order $<_\mu$ for any integer $\mu \geq 0$. We choose
$\mu=0$.

Let $\errs = |E|$ be the number of errors, unknown to the receiver. Then $\deg
\Omega < \deg \Lambda = \errs$. As $\deg \Lambda > \deg \Omega$, then
$y\Lambda(x) >_0 \Omega(x)$. Assume now that $\ydeg_\lz h_2 = 1$ (and therefore
$\ydeg_\lz h_1 = 0$). Therefore, by Proposition
\ref{prop:grobDegs}, we know there exist polynomials $f_1, f_2 \in \F[x]$ such
that
\begin{align}
    \Omega(x) - y\Lambda(x) &= f_1(x)h_1(x,y) + f_2(x)h_2(x,y) \notag \\ 
    \deg f_1 &\leq \errs - d + \xdeg_\lz(h_2) \label{eqn:RSfdeg}\\
    \deg f_2 &= \errs - \xdeg_\lz(h_2) \notag
\end{align}
We see that whenever $\errs \leq \lfloor \frac
{n-k} 2 \rfloor$, either the degree bound for $f_1$ or that for $f_2$ will be
negative, and that one will then be zero. Therefore $\Omega(x) - y\Lambda(x)$
will be a multiple of either $h_1$ or $h_2$. As $\ydeg_\lz(\Omega(x) -
y\Lambda(x)) = 1$, it must be a multiple of $h_2$. However, as $\Lambda$ and
$\Omega$ are coprime, that multiple must be the constant that normalises $h_2$
to have leading coefficient 1, just as $\Lambda(x)$. This corresponds to the
Sugiyama decoding algorithm \cite{sugiyama}.

In case neither $h_1$ nor $h_2$ is valid as $\Omega(x)-y\Lambda(x)$, we know
that $f_1$ and $f_2$ are non-zero, so there are more errors than half the
minimum distance; then we proceed exactly like regular Wu list decoding using
BMA. We know that for at least $\errs$ values of $x_0 \in \{ \alpha_0, \ldots,
\alpha_{n-1} \}$, we have $\Lambda(x_0) = 0$, namely the error locations.
Therefore, by \eqref{eqn:RSfdeg}, for at least those $\errs$ values of $x_0$,
we have $f_1(x_0)h_{11}(x_0) + f_2(x_0)h_{21}(x_0) = 0$. Thus, for this to be a
rational interpolation problem as in Section \ref{sec:ratInter}, we just need
to ascertain two properties: 1) that $h_{11}(x)$ and $h_{21}(x)$ never
simultaneously evaluate to zero since they are coprime, as a linear combination
of $h_1$ and $h_2$ equals $y - S(x) \in M$; 2) that $f_1$ and $f_2$ are coprime since $\Lambda$ and $\Omega$ are.

From the results developed in Section \ref{sec:ratInter}, we can therefore
solve this rational interpolation problem for certain values of $\ell$ as well
as the parameters $n$ and $d$: we
construct a partially homogeneous interpolation polynomial $Q(x,y,z)$ which has
zero at all the points $(\alpha_i, h_{11}(\alpha_i), h_{21}(\alpha_i))$ for
$i=0,\ldots,n-1$. Under certain constraints on the degrees of $Q(x,y,z)$, then
$yf_1(x) + zf_2(x)$ will be a factor of $Q(x,y,z)$. The following subsection
looks closer at the possible choice of parameters to derive the
upper bound on $\tau$. The complete list decoder is listed in Algorithm
\ref{alg:rs}. 

\Remark There is a duality between the GSA and the Wu
list decoder: in list decoding GRS codes with the GSA,
one sets up an interpolation problem where the sought solution -- the
information word -- will pass through those of the prescribed points that
correspond to the error-free positions. Oppositely, here we seek $f_1, f_2$
that pass through those of the prescribed points that correspond to the errors
positions.
\RemarkEnd

\subsection{Analysis of the parameters}\label{sec:grsParams}

It is clear that in Theorem \ref{thm:interpol}, we should set $w_1,w_2$ equal
to the bounds on $\deg f_1, \deg f_2$ in \eqref{eqn:RSfdeg} for the case $\eps
= \tau$; so $w_1+w_2 = 2\tau - d$.  The main question is then for which $\tau$ we can select
$\ell$ and $s$ such that \eqref{eqn:ratPermissible} is satisfied. Inserting
the value for $w_1+w_2$ and rearranging, \eqref{eqn:ratPermissible} becomes
\begin{equation}\label{eqn:grsPermissible}
  \frac \tau n < \frac 1 {(\ell+1)(\ell-s)}
  \left( \binom {\ell+1} 2 \frac d n - \binom {s+1} 2 \right)
\end{equation}
Replacing $s$ by $\ell-s$ this is exactly the equation governing the choice of
parameters $s, \ell$ and $\tau$ in the GSA for the same values of $n$ and $d$,
see e.g. \cite[Lemma 9.5]{rothBook}. This means that for all parameters of the
GSA where the multiplicity is less than $\ell$, this substitution applies,
giving valid parameters for Algorithm 1.\begin{footnote}{%
  We are grateful to the anonymous reviewer for pointing out this relation to
  us.}
We arrive at the following two lemmas:
\begin{lemma}
Algorithm \ref{alg:rs} can list decode for any $\tau < n - \sqrt{n(n-d)}$.
\end{lemma}
\begin{proof}
Assume $\tau \geq \frac d 2$  since otherwise the minimum distance decoding done in steps 1--3 is sufficient.
We first wish to show that we can select parameters $s, \ell$ such that \eqref{eqn:grsPermissible} is fulfilled.

For any $\tau$ less than the bound of the lemma, there exists a valid list size $\ell$ and multiplicity $s_G$ such that the equation of the GSA is satisfied, and furthermore $s_G \leq \ell$, see e.g.~\cite[Lemma 9.5]{rothBook}.
Except in the case $s_G = \ell$, the duality between Algorithm~\ref{alg:rs} and the GSA applies and we can choose $s = \ell-s_G$.
However, if $s_G = \ell$, since the governing equation of the GSA is satisfied, we have $\frac \tau n < \frac 1 {(\ell+1)\ell}\binom {\ell+1} 2 \frac d n$, but this contradicts $\tau \geq \frac d 2$.

We can therefore choose $s, \ell$ such that \eqref{eqn:grsPermissible} is fulfilled; therefore whenever $\errs = \tau$, Theorem \ref{thm:interpol} promises that the sought $f_1, f_2$ will be found in step 5 of Algorithm~\ref{alg:rs}.

Now, to be guaranteed to find them also whenever $\errs < \tau$, we need to employ
Lemma \ref{lem:risingRat}. This can be used if it is satisfied
that
\begin{align*}
\min\{w_1 - \deg f_1, w_2 - \deg f_2\} &\geq \tfrac s \ell (\tau - \errs)
\end{align*}
Note that $w_1 - \deg f_1 \geq \tau - \errs$ using \eqref{eqn:RSfdeg}. The
same holds for $w_2 - \deg f_2$. Therefore, the above is true at least if we
satisfy
\begin{IEEEeqnarray*}{rCl?C?rCl}
\tau - \errs &\geq& \tfrac s \ell (\tau - \errs)
    & \iff &
s &\leq& \ell
\end{IEEEeqnarray*}
Thus, Lemma \ref{lem:risingRat} guarantees that as long as $s \leq \ell$, then
the $Q(x,y,z)$ we construct in step 4 will contain $yf_1(x)+zf_2(x)$ as a factor whenever $\errs
\leq \tau$. But $s \leq \ell$ is satisfied as $0 < s_G < \ell$ in all considered
cases of the GSA and $s = \ell-s_G$.
\end{proof}

\Remark This decoding radius -- the so-called Johnson bound -- is not the best
one can achieve for a given GRS code: using the GSA+KV one can decode
slightly further, namely up to the $q$-ary
Johnson bound $\frac {q-1} q\big(n - \sqrt{n(n-\frac q {q-1}d})\big)$, see e.g.
\cite{augot10} or \cite[Section 9.6]{rothBook}. 
\RemarkEnd

\begin{lemma}\label{lem:grsParamEquiv}
For given $n, k$ and $\tau$ with $\tau \geq \frac {n-k+1} 2$, then $\ell$ and
$s$ are valid choices for the
parameters for Algorithm~\ref{alg:rs} if and only if $\ell$ and $s_G = \ell-s$
are valid choices for the GSA. Furthermore, for any given $\ell$, let $s$ be
the smallest possible choice of multiplicity for Algorithm \ref{alg:rs} and $s_G$
the smallest possible choice of multiplicity for the GSA; if $\tau < n/2$ then $s
\leq s_G$, otherwise, $s \geq s_G$.
\end{lemma}
\begin{proof}
Only the last claim does not directly follow from the duality in parameter
choice. Consider \eqref{eqn:grsPermissible} governing the possible choice of
$s$ for Algorithm \ref{alg:rs}: rearranging to a second-degree equation in $s$
and solving, we get that $s/\ell$ must be chosen from the interval
$[T-\sqrt D;\ T+\sqrt D]$, where $T = \frac \tau n + \frac {\tau -n/2} {n\ell}$
and $D$ a discriminant whose precise expression is not important for us. Due to
the duality between Algorithm \ref{alg:rs} and the GSA, the corresponding
interval for valid $s_G/\ell$ for the GSA will be $[1-T-\sqrt D;\ 1-T+\sqrt
D]$. In addition to residing in these respective intervals, we only require of
$s/\ell$ and $s_G/\ell$ that $s$ and $s_G$ are positive integers less than
$\ell$. Therefore, whenever $\tau < n/2$ we have $T < \half$, so the smallest 
possible choice of $s$ in the former interval must be at most the smallest 
possible in the latter interval; oppositely for the case $\tau \geq n/2$.
\end{proof}

To concretely choose $\ell$ and $s$ given $n,k$ and $\tau$, we can---due to
the above lemma---use closed expressions designed for the GSA; e.g.
\cite[Eqs.(43-45)]{mceliece03}. Alternatively, Trifonov and Lee give a simple
analysis and expressions directly for the Wu list decoder in \cite{trifonov12}.

\begin{algorithm}
  \caption{Wu list decoding GRS codes}
  \label{alg:rs}
  \begin{algorithmic}[1]
    \Require A GRS code $\C$ over $\FF q$ with parameters $n, k, d=n-k+1$ and
             evaluation points $\alpha_0,\ldots,\alpha_{n-1}$, decoding radius
             $\tau < n-\sqrt{n(n-d)}$, and received word $r \in \FF q^n$.
    \Ensure  A list of all codewords in $\C$ within radius $\tau$ of $r$ or
             \Fail if there are no such words.
    \Statex 
    \State Calculate the syndrome $S(x)$ from $r$ according to
           \eqref{eqn:rsSyn}. 
    \State Run the EA on $x^{d-1}, S(x)$ and halt when $\deg s_i < \deg v_i$,
           reusing the notation of Section \ref{sec:ea}. Define $\tilde h_1(x)
           = -v_{i-1}(x)$ and $\tilde h_2(x) = -v_i(x)$.
    \State If $\tilde h_2$ is a valid error-locator of degree at most $d-\tau$,
           use it to correct $r$, and
           if this yields a word in $\C$, return this one word.
    \State Otherwise, we seek $f_1, f_2$ according to \eqref{eqn:RSfdeg}.
           Set $w_1, w_2$ to the degree bounds of $f_1$ and $f_2$ for the case
           $\errs = \tau$, and calculate $\ell$ and $s$ to satisfy
           \eqref{eqn:grsPermissible}. Construct a $Q(x,y,z)$ satisfying the
           requirements of Theorem \ref{thm:interpol} using the points
           $\{(\alpha_i, \tilde h_1(\alpha_i), \tilde h_2(\alpha_i))
           \}_{i=0}^{n-1}$.
    \State Find all factors of $Q(x,y,z)$ of the form $yf_1\T(x) + zf_2\T(x)$
           where $f_1\T$ and $f_2\T$ have degree less than $w_1$ and $w_2$
           respectively. Return \Fail if no such factors exist.
    \State For each such factor, construct $\Lambda\T(x) = f_1\T(x)\tilde
           h_1(x) + f_2\T(x)\tilde h_2(x)$. If it is a valid error-locator, use
           it for correcting $r$. Return \Fail if none of the factors yield
           error-locators
    \State Return those of the corrected words that are in $\C$. Return \Fail
           if there are no such words.
  \end{algorithmic}
\end{algorithm}

\subsection{Complexity analysis}\label{sec:grsComplexity}

The complexity of the totality of Algorithm \ref{alg:rs} is easily found using
the results of Section \ref{sec:ratFast}; note that $\tau^2 > n(w_1+w_2)$ whenever
$\tau < n-\sqrt{n-d}$ so we can use Lemma \ref{lem:interpolComp}. For
simplicity, we will assume that
$q \in O(n)$ where $q$ is the cardinality of $\F$. In that case, as $\ell \geq
s$, steps 4 and 5 can be computed in $O(\ell^{\omega+1}n\log^{O(1)}(\ell n))$.
The remaining steps are of lower order: calculating $S(x)$ in step 1 can be
done in $O(n\log n)$ using fast Fourier methods, and the EA in step 2 has
complexity $O(n\log^2 n\log\log n)$. Checking whether a polynomial is a valid
error-locator takes at most $O(q)$, and in step 3 we check two such, while in
step 6 we check at most $\ell$ such. Thus we have the following

\begin{lemma}\label{lem:rsComplexity} 
  If $q \in O(n)$ then Algorithm \ref{alg:rs} has complexity
  $O(\ell^\omega s n\log^{O(1)}(\ell n))$.
\end{lemma}

Using Lemma \ref{lem:grsParamEquiv} we can compare running times with those for
variants of the GSA. In this light, the above estimate is fast as the fastest
GRS
list-decoders based on the GSA. The bottle-neck is -- as it is here -- the
construction of an interpolation polynomial. Beelen and Brander gave in
\cite{beelenBrander} an algorithm for computing the interpolation polynomial in
the GSA with complexity $O(\ell^5n\log^2 n\log\log n)$, using an approach very
close to the one here, and using a row reduction algorithm on an appropriate
polynomial matrix. However, had they used the one by Giorgi et al.
\cite{giorgi03} instead of the slightly slower by Alekhnovich
\cite{alekhnovich05}, they would have reached the same complexity as in Lemma
\ref{lem:rsComplexity}, but using the value of $s$ needed for the GSA.

It would therefore seem that, when the multiplicity $s$ for Algorithm \ref{alg:rs}
is smaller than the multiplicity $s_G$ for the GSA, Algorithm \ref{alg:rs} would be
faster than the GSA, though as we have only presented asymptotic analysis, one
would need implementations to properly verify this. From Lemma
\ref{lem:grsParamEquiv} and its proof, we know that the multiplicity for
Algorithm \ref{alg:rs} is smaller whenever $\tau < n/2$ and that the
difference from the multiplicity of the GSA increases with decreasing $\frac \tau n$.

Bernstein also gives a decoding algorithm in \cite{bernstein11simplified} with
the same complexity, but his is a variant of the GSA+KV, and it can thus decode
a GRS code to the slightly higher $q$-ary Johnson radius: $\frac{q-1} q\big(n -
\sqrt{n(n - \frac q {q-1}d)}\big)$; see also Section \ref{sec:goppaComplexity}.

\endgroup 
\section{Wu list decoding binary Goppa codes}\label{sec:goppa}
\begingroup 
\newcommand{\lo}{{<_1}}

\subsection{The codes}

Consider an irreducible polynomial $g(x) \in \FF{2^m}[x]$ as well as $n$
distinct elements of $\FF{2^m}$, $L = (\alpha_0, \ldots, \alpha_{n-1})$.
Then
the irreducible binary Goppa code $\Gamma(g, L)$ with Goppa polynomial $g$ over
$L$ is the set
\[
     \left\{ (c_1, \ldots, c_n) \in \F_2^n   \ \Bigg|\ 
        \sum_{i=0}^{n-1} \frac {c_i}{x - \alpha_i} \equiv 0 \mod g(x) \right\}
\]
This code has parameters $\params n {\geq n - m\deg g} {\geq 2\deg g+1}$. A
binary Goppa code $\Gamma(g,L)$ is a subfield subcode of an $\params n {n-\deg
g} {\deg g+1}$ GRS code over $\FF{2^m}$. It is also an alternant code. See e.g.
\cite{MacWilliamsSloane} for a more complete description.

Consider a sent codeword $c = (c_0, \ldots, c_{n-1})$ and a corresponding
received word $r = (r_0, \ldots, r_{n-1})$. For these codes, a natural
definition of a syndrome polynomial is then
\begin{equation}\label{eqn:gopSyn}
  S(x) = \left( \sum_{i=0}^{n-1} \frac {r_i}{x - \alpha_i} \modop g(x) \right)
\end{equation}
Like in the preceding section, we also define $E$, the error-locator and
error-evaluator, the last being slightly simpler due to the binary field:
\begin{IEEEeqnarray*}{rCl}
    E          &=& \{ i \mid c_i \neq r_i \} \\
    \Lambda(x) &=& \prod_{i \in E} (x-\alpha_i) \\
    \Omega(x)  &=& \sum_{i \in E} \prod_{j \in E \setminus \{i\}} (x-\alpha_j)
\end{IEEEeqnarray*}
Introduce also $\errs = |E|$ as the number of errors.
We have again that $\gcd(\Lambda, \Omega) = 1$. It also turns out that
the introduced polynomials satisfy a Key Equation \cite{patterson75}:
\begin{equation}
    \Lambda(x)S(x) \equiv \Omega(x) \mod g(x)
\end{equation}
Note that for a binary code, the receiver can decode immediately upon having
calculated the error locator, even without the error evaluator; the error value
is always 1.

\subsection{The list-decoding algorithm}
Now we could proceed exactly as in Section \ref{sec:RSdecode}, and we would
arrive at a list decoder correcting up to $n - \sqrt{n(n-\deg g-1)}$ errors.
This is the same decoding radius reached by simply decoding the enveloping GRS
code with the GSA or the Wu decoder. However, this radius
is much less than $\deg g$ which is promised by the minimum distance of the
binary Goppa code, and which can be corrected by Patterson's decoder
\cite{patterson75}.

Therefore, we proceed to rewrite the Key Equation in the same way as Patterson.
In the following, it will be useful to assume $\errs < 2\deg g$ as an initial
and reasonable bound on our list decoder. Then, collecting even and odd terms,
we can introduce polynomials $a(x), b(x)$ such that $\Lambda(x) = a^2(x) +
xb^2(x)$ and satisfying $\deg a \leq \lfloor \frac \errs 2 \rfloor$ and $\deg b
\leq \lfloor \frac {\errs-1} 2 \rfloor$. Now, note from the definition of the
polynomials that $\Omega(x)$ equals the formal derivative of $\Lambda(x)$, so
we get $\Omega(x) = b^2(x)$ in this field of characteristic 2. The Key Equation
thus becomes
\begin{IEEEeqnarray}{rCll/r}
  (a^2(x) + xb^2(x))S(x) &\equiv& b^2(x) &\mod g(x)    &\iff \notag\\
  b^2(x)(x + S^{-1}(x))  &\equiv& a^2(x) &\mod g(x)    \label{eqn:goppaDerv}
\end{IEEEeqnarray}
Note here that calculating the inverse of $S(x)$ modulo $g(x)$ is possible
since $\deg S < \deg g$ and $g(x)$ is irreducible. 

It might now be that $S^{-1}(x) \equiv x \mod g(x)$ in which case $a^2(x)
\equiv 0 \mod g(x)$. As $g(x)$ is irreducible, $a(x)$ must be a multiple of
$g(x)$, which means that $a(x) = 0$ as $\errs < 2\deg g$. This implies
$\Lambda(x) = xb^2(x)$, which is only a legal error locator if $0 \in L$ and
$b(x) = 1$. So in that case, $\Lambda(x) = x$ is the only valid solution to the
Key Equation, resulting in one error to be corrected.

Having taken care of the case $S^{-1}(x) \equiv x \mod g(x)$, let us now assume
that this is not the case and continue. As $g(x)$ is irreducible,
$\F_{2^m}[x]/\langle g(x) \rangle$ is a finite field of characteristic 2, so we
can compute a square-root; in particular, we can find an $\tilde S(x)$ such that
$\tilde S^2(x) \equiv x + S^{-1}(x) \mod g(x)$ and $\deg \tilde S < \deg g$.
This value is directly computable by the receiver after having computed $S(x)$.
Inserting $\tilde S(x)$ in \eqref{eqn:goppaDerv}, we get
\begin{IEEEeqnarray}{rCll/r}
  b^2(x)\tilde S^2(x)  &\equiv& a^2(x) &\mod g(x)    &\iff \notag \\
  b(x)\tilde S(x)  &\equiv& a(x) &\mod g(x) \label{eqn:gopKeyab}
\end{IEEEeqnarray}
Now we are in the case of a new Key Equation, where the degrees of the unknown
polynomials are halved! We proceed in a manner resembling that of the GRS codes
from the preceding section. The above equation tells us that $a(x) - yb(x) \in
M = [g(x), y-\tilde S(x)]$. If we run the EA on $g(x)$ and $\tilde S(x)$, by
Proposition
\ref{prop:EAgrob}, we get a Gr\"obner basis $G = \{h_1, h_2\}$ of $M$ of module
term order $<_\mu$ for any integer $\mu \geq 0$; for
reasons becoming apparent momentarily, we choose $\mu = 1$.

By
Proposition \ref{prop:grobDegs}, we know there exist polynomials $f_1, f_2 \in
\F[x]$ such that
\begin{equation}\label{eqn:gopSplitab}
    a(x) - yb(x) = f_1(x)h_1(x,y) + f_2(x)h_2(x,y)
\end{equation}
Remembering Proposition \ref{prop:moduleGrob}, assume that $\ydeg_\lo h_2 = 1$
and therefore that $\ydeg_\lo h_1=0$. Now, the case here is slightly more
complicated than that of the GRS codes, as we do not know a priori which of
$a(x)$ and $b(x)$ has the largest degree. If $\errs$ is even then $\deg a =
\frac \errs 2$ and $\deg b \leq \frac \errs 2 - 1$ whereby $a(x) >_1 yb(x)$.
From Proposition \ref{prop:grobDegs} we then get
\begin{IEEEeqnarray}{rClCl}
    \deg f_1 &=& \deg a - \deg_\lo^x(h_1) 
             &=& \tfrac \errs 2 - \deg g + \deg_\lo^x(h_2)
                \notag \\[2pt]
    \deg f_2 &\leq& \deg a - 1 - \deg_\lo^x(h_2) 
             &=& \tfrac \errs 2 - 1 - \deg_\lo^x(h_2)
    \label{eqn:gopfdeg}
\end{IEEEeqnarray}
In a similar manner, when $\errs$ is odd we get $a(x) <_1 yb(x)$ and 
\begin{IEEEeqnarray}{rClCl}
    \deg f_1 &\leq& \deg b + 1 - \deg_\lo^x(h_1) - 1
             &=& \tfrac {\errs-1} 2 - \deg g + \deg_\lo^x(h_2)
                \notag \\[2pt]
    \deg f_2 &=& \deg b  - \deg_\lo^x(h_2) 
             &=& \tfrac {\errs-1} 2 - \deg_\lo^x(h_2)
    \label{eqn:gopfdegOdd}
\end{IEEEeqnarray}
In either of the above cases, we see that if $\errs \leq \deg g$, one of the
bounds for $\deg f_1$ and $\deg f_2$ will be negative, in which case either
$f_1$ or $f_2$ will be zero. This in turn means that $a(x) - yb(x)$ will be a
multiple of either $h_1$ or $h_2$, namely the one which has the same $y$-degree
as $a(x) - yb(x)$ under $<_1$. As $\Lambda(x)$ is square-free, $a(x)$ and
$b(x)$ must be
relatively prime, so this multiple must be a constant. This corresponds to
Patterson's decoder \cite{patterson75}, except that there the BMA is used
instead of the EA to solve \eqref{eqn:gopKeyab}. This requires an initial
transformation of \eqref{eqn:gopKeyab}, and an ``inverse'' transformation on
the output of the BMA.

In case $f_1$ and $f_2$ are both non-zero, spurred on by the success of the
last section, we would like to be able to find them using rational
interpolation. However, in the last section, we knew that the evaluation of the
target polynomial $\Lambda(x)$ would be 0 in at least $\errs$ positions; for
neither $a(x)$ nor $b(x)$ do we have such information. We therefore first need
to re-enter \eqref{eqn:gopSplitab} into their defining expression: $\Lambda(x)
= a^2(x) + xb^2(x)$. Let first $h_1(x,y) = h_{10}(x) + yh_{11}(x)$ and
$h_2(x,y) = h_{20}(x) + yh_{21}(x)$. Then using \eqref{eqn:gopSplitab}, we get
\begin{IEEEeqnarray*}{rCll}
  \Lambda(x) &=&  \ &(f_1(x)h_{10}(x) + f_2(x)h_{20}(x))^2 \\
             & & + x&(f_1(x)h_{11}(x) + f_2(x)h_{21}(x))^2  \\
             &=& \IEEEeqnarraymulticol{2}{l}{
                    f_1^2(x)(h_{10}^2(x)+xh_{11}^2(x))
                     + f_2^2(x)(h_{20}^2(x)+xh_{21}^2(x))
                  }
\end{IEEEeqnarray*}
Similarly to the preceding section, for at least $\errs$ values of $x_0 \in L$,
we now know that $\Lambda(x_0) = 0$. For these $\errs$ values of $x_0$, by the
above, we therefore have
\[
  f_1(x_0)\sqrt{\hat h_1(x_0)} + f_2(x_0)\sqrt{\hat h_2(x_0)} = 0
\]
where $\hat h_1(x) = h_{10}^2(x)+xh_{11}^2(x)$ and $\hat h_2(x) =
h_{20}^2(x)+xh_{21}^2(x)$. For us to be able to use Theorem \ref{thm:interpol},
we have then only to certify that $f_1$ and $f_2$ are coprime, and that $\hat
h_1$ and $\hat h_2$ will never simultaneously evaluate to zero. But the former
is true since $a$ and $b$ are coprime which is due to $\Lambda$ being
square-free, and the latter is true since $h_1(x,y)$ and $h_2(x,y)$ are
coprime. We have therefore finally arrived at a rational interpolation problem. 

We will again use the results of Section \ref{sec:ratInter} to solve this
problem for some values of $\errs, n, \deg g$. The next section is concerned
with that analysis. The complete list decoder is shown in Algorithm
\ref{alg:goppa}.

\Remark
As mentioned, Patterson's original algorithm \cite{patterson75} solves
\eqref{eqn:gopKeyab} using the BMA. One could possibly also extend this for
list decoding using rational interpolation. However, a transformation is needed
for letting the BMA solve \eqref{eqn:gopKeyab}, and this makes the details for
rational interpolation less
straight-forward. One should also note that the BMA and the EA in their
straightforward implementations have the same asymptotic running time
$O(\theta^2)$ (see e.g.
\cite{dornstetter87}), and that both admit a recursive version with asymptotic
running time $O(\theta\log^2 \theta\log\log \theta)$, where $\theta$ is the degree of the
ingoing polynomials (see e.g. \cite[Chapter 11.7]{blahut83}
respectively \cite[Chapter 8.9]{ahoDesignCS}).
\RemarkEnd

\subsection{Analysis of the parameters}

For a given
decoding radius $\tau$, we want to know whether we can construct a $Q(x,y,z)$
such that whenever $\errs \leq \tau$, we can find $f_1$ and $f_2$ in the manner
specified in Theorem \ref{thm:interpol}, and we want values for the parameters
of $\ell$ and $s$. 

We should set $w_1, w_2$ inspired by \eqref{eqn:gopfdeg} and
\eqref{eqn:gopfdegOdd}, but we need just one set of values which will cover
both the even and odd cases. Therefore, we use for both $f_1$ and $f_2$ the
larger of the degree bounds:
\begin{IEEEeqnarray}{rCl}
    w_1 &=& \tfrac \tau 2 - \deg g + \deg_\lo^x(h_2) \notag \\
    w_2  &=& \tfrac {\tau-1} 2 - \deg_\lo^x(h_2) \label{eqn:gopwset}
\end{IEEEeqnarray}
Now define $w = w_1 + w_2 = \tau - \deg g - \half$. Note that $w$ and either
$w_1$ or $w_2$ will not be integer. Inserting the value for $w$ and
rearranging, \eqref{eqn:ratPermissible} becomes
\begin{equation}\label{eqn:gopPermissible}
\!
  \frac \tau n < \frac 1 {(\ell+1)(\ell-2s)}
  \left( \binom {\ell+1} 2 \frac {\deg g + \frac 1 2} n -2\binom {s+1} 2 \right)
\!
\end{equation}
if we assume that $\ell > 2s$. Just as we before found that the governing
equation for Algorithm \ref{alg:rs} is parallel to that of the GSA, the above
equation is parallel to the governing equation of the GSA+KV: using e.g.
\cite[Lemma 9.7]{rothBook} and setting the two multiplicities as $r=\ell-s$ and
$\bar r=s$ we achieve the same equation. This means that Algorithm
\ref{alg:goppa} has the same decoding radius as the GSA+KV when the choice of
the two multiplicities are restricted thusly. From \cite[Problem
9.9]{rothBook}, the choice $\bar r = \ell-r$ exactly maximises the decoding
radius which is then given in \cite[Problem 9.10]{rothBook}. We also get $\bar
r < r$ so $\bar r < \ell/2$ and hence in our case $\ell >
2s$; this is also what we assumed in \eqref{eqn:gopPermissible} which means we
can indeed reuse the analysis from the GSA+KV.

\begin{lemma}\label{lem:goppaRadius}
Algorithm \ref{alg:goppa} can list decode for any $\tau < \half n- \half
\sqrt{n(n-4\deg g-2)}$.
\end{lemma}
\begin{proof}
With the above duality, we have already established that for any given $\tau$
less than the given decoding radius we can select values of $s$ and $\ell$ such
that the sought $f_1$ and $f_2$ can be found whenever $\errs = \tau$. We again
have to employ Lemma \ref{lem:risingRat} in order to guarantee that $f_1$ and
$f_2$ will be found when $\errs < \tau$. The lemma promises this if we can
satisfy
\[
  \min\{w_1 - \deg f_1, w_2 - \deg f_2\} \geq \tfrac s \ell (\tau - \errs)
\]
Using \eqref{eqn:gopfdeg}, \eqref{eqn:gopfdegOdd}  and \eqref{eqn:gopwset}, we
see that $w_1 - \deg f_1 \geq \tfrac \tau 2 - \lfloor \tfrac \errs 2 \rfloor
\geq \half(\tau - \errs)$, both when $\errs$ is even and when it's odd.
Similarly for $w_2 - \deg f_2$. The condition of Lemma \ref{lem:risingRat} is
then always satisfied
as long as $\ell > 2s$. This we already assumed in \eqref{eqn:gopPermissible}.
\end{proof}

\Remark
It is the necessity of having to use Lemma \ref{lem:risingRat} that adds the
peculiar complication
on the setting of $w_1$ and $w_2$. If we choose a $\tau$, we will know its
parity, so we could choose $w_1$ and $w_2$ as implied from \eqref{eqn:gopfdeg} or
\eqref{eqn:gopfdegOdd}, according to that parity. This would allow us to decode
exactly $\tau$ errors; analysis shows that in that case one could choose any
$\tau <
\half n - \half\sqrt{n(n-4\deg g-4)}$, i.e. slightly greater than the binary
Johnson radius. However, the condition of Lemma \ref{lem:risingRat} would
then not always be true so we would not always be able to correct \emph{fewer}
errors. This is the reason of having to set $w_1$ and $w_2$ as in
\eqref{eqn:gopwset}.

Interestingly, if we allow two runs of the rational interpolation procedure
instead of just one, we \emph{can} achieve the decoding radius $\tau <
\half n - \half\sqrt{n(n-4\deg g-4)}$ and still also decode fewer than $\tau$
errors: let the first run be responsible for finding those error locators
corresponding to even number of errors, and the second run for the odd number
of errors. For each run we then only need a looser version of Lemma
\ref{lem:risingRat}, where only a number of points with the right parity need
to be interpolated as well. Then we can set $w_1, w_2$ according to
\eqref{eqn:gopfdeg} in the even-parity run, and similarly $w_1, w_2$ from
\eqref{eqn:gopfdegOdd} in the odd-parity run. This yields the mentioned
decoding radius.
\RemarkEnd

\begin{lemma}
For given $n$, $\deg g$ and $\tau$, then $\ell$ and $s$ are valid choices for
the parameters for Algorithm \ref{alg:goppa} if and only if $\ell$, $r=\ell-s$
and $\bar r=s$ are valid parameters for the GSA+KV as described in
\cite[\S 9.6]{rothBook}.
\end{lemma}

For closed expressions for valid values of the parameters $\ell$ and $s$, one
can use the analysis of Trifonov and Lee \cite{trifonov12} which works for any
application of the rational interpolation method.

\begin{algorithm}
  \caption{Wu list decoding binary Goppa codes}
  \label{alg:goppa}
  \begin{algorithmic}[1]
    \Require A binary Goppa code $\C$ with Goppa polynomial $g(x) \in \FF
      {2^m}[x]$ and evaluation points $\alpha_0,\ldots,\alpha_{n-1}$, a
      decoding radius $\tau < \half n-\half \sqrt{n-4\deg g-2}$, and a received
      word $r \in \FF 2^n$.
    \Ensure  A list of all codewords in $\C$ within radius $\tau$ of $r$ or
             \Fail if there are no such words.
    \Statex 
    \State Calculate the syndrome $S(x)$ from $r$ according to
           \eqref{eqn:gopSyn}. If $S^{-1}(x) = x$ and $0 \in L$, then flip the
           corresponding bit of $r$ and return that word. If $S^{-1}(x) = x$
           and $0 \notin L$, return \Fail. Otherwise, calculate $\tilde S(x)$
           satisfying $\deg \tilde S < \deg g$ and $\tilde S^2(x) \equiv x +
           S^{-1}(x) \mod g(x)$.
    \State Run the EA on $g(x), \tilde S(x)$ and halt when $\deg s_i < \deg
           v_i+1$,
           reusing the notation of Section \ref{sec:ea}. Define $\hat h_1(x) =
           s_{i-1}^2(x) + xv_{i-1}^2(x)$ and $\hat h_2(x) = s_i^2(x) +
           xv_i^2(x)$.
    \State If either $\hat h_1(x)$ or $\hat h_2(x)$ are valid error-locators of
           degree at most $2\deg g-\tau$, use
           that to decode, and if this yields a word in $\C$, return this one
           word.
    \State Otherwise, we seek $f_1, f_2$ according to \eqref{eqn:gopSplitab}.
           Set $w_1, w_2$ as in \eqref{eqn:gopwset}, and calculate $\ell$ and
           $s$ to satisfy \eqref{eqn:gopPermissible}. Construct a $Q(x,y,z)$
           satisfying the requirements of Theorem \ref{thm:interpol} using the
           points $\big\{\big(\alpha_i, \sqrt{\hat h_1(\alpha_i)}, \sqrt{\hat
           h_2(\alpha_i)}\big)
           \big\}_{i=0}^{n-1}$.
    \State Find all factors of $Q(x,y,z)$ of the form $yf_1\T(x) + zf_2\T(x)$
           where $f_1\T$ and $f_2\T$ have degree less than $w_1$ and $w_2$
           respectively. Return \Fail if no such factors exist.
    \State For each such factor, construct $\Lambda\T(x) = f_1\T{}^2(x)\hat
           h_1(x) + f_2\T{}^2(x)\hat h_2(x)$. If it is a valid error-locator,
           use it for decoding $r$. Return \Fail if none of the factors yield
           error-locators
    \State Return those of the decoded words that are in $\C$. Return \Fail if
           there are no such words.
  \end{algorithmic}
\end{algorithm}

\subsection{Complexity Analysis}\label{sec:goppaComplexity}

Again, the complexity of Algorithm \ref{alg:goppa} is easily found using the
results of Section \ref{sec:ratFast}. For simplicity, we will assume that $2^m
\in O(n)$. In that case, as $\ell \geq s$, steps 4 and 5 can be computed in
$O(\ell^\omega s n\log^{O(1)}(\ell n))$. The remaining steps are of lower
order, seen using arguments similar to those in Section
\ref{sec:grsComplexity}.

\begin{lemma}\label{lem:gopComplexity} 
  If $2^m \in O(n)$ then Algorithm \ref{alg:goppa} has complexity
  $O(\ell^\omega s n\log^{O(1)}(\ell n))$.
\end{lemma}

The GSA+KV can decode
binary Goppa codes -- in fact any alternant code -- up to the small-field
Johnson bound. Also here, the bottle-neck of the complexity is the construction
of the interpolation polynomial. Bernstein in \cite{bernstein11simplified}
gives an algorithm for constructing this fast, and in terms of $\ell$ and $n$
and relaxing $s, r$ and $\bar r$ to $\ell$, it has the same complexity as the
above.

However, similarly to Section \ref{sec:grsComplexity}, one should note that
$s=\bar r < \ell/2$ and $r =\ell-\bar r > \ell/2$, and the difference between
$s$ and $r$ increases with the rate of the code. From this view, one would
therefore expect that Algorithm \ref{alg:goppa} outperforms the GSA+KV, though
the asymptotic analysis we have performed here is too crude to say for certain.

\endgroup

\section{Conclusion}

In this article, we have reinvestigated the Wu list decoder of \cite{wu08}.
Originally formulated in tight integration with the Berlekamp-Massey algorithm,
we have shown how the extended Euclidean algorithm can be used instead,
enabling one to solve more general equations than the original Key Equation for
Generalised Reed-Solomon codes. 

At its core, the Wu list decoder solves a rational interpolation problem in a
manner mirroring the polynomial interpolation of the Guruswami-Sudan algorithm
(GSA). We have pointed out how this equation becomes the one of the GSA by a
change of variables, implying that their decoding radii and list sizes are the
same, as well as connecting the multiplicities. 

The most expensive part of solving the rational interpolation problem is the
construction of an interpolation polynomial. We have shown how to extend
methods used in the GSA for constructing this polynomial fast. The result is
that the Wu list decoder can be made to run in the same complexity as the
fastest variants of the GSA.

The decoupling of the Key Equation-solving and rational interpolation from
the actual decoding results in a short derivation of the list decoder for GRS
codes. Moreover, it makes it clear that the approach also can be used to extend
the Patterson decoder for binary Goppa codes, list decoding up to the binary
Johnson radius. Also here, a connection to the governing equation of the GSA
with the K\"otter-Vardy multiplicity assignment method is pointed out.

\section{Acknowledgements}
The authors are very grateful for the insightful comments and suggestions made
by the anonymous reviewers.

\bibliographystyle{IEEEtran}
\bibliography{../tex/bibtex}

\end{document}